\documentclass[pdflatex,sn-mathphys-num]{sn-jnl}

\usepackage{graphicx}%
\usepackage{multirow}%
\usepackage{amsmath,amssymb,amsfonts}%
\usepackage{amsthm}%
\usepackage{mathrsfs}%
\usepackage[title]{appendix}%
\usepackage{xcolor}%
\usepackage{textcomp}%
\usepackage{manyfoot}%
\usepackage{booktabs}%
\usepackage{algorithm}%
\usepackage{algorithmicx}%
\usepackage{algpseudocode}%
\usepackage{listings}%
\usepackage{pifont} 
\usepackage{float}
\usepackage{xurl}

\theoremstyle{thmstyleone}%
\newtheorem{theorem}{Theorem}
\newtheorem{corollary}{Corollary}
\theoremstyle{thmstyletwo}%

\theoremstyle{thmstylethree}%

\begin{document}

\title[Article Title]{String stable platoons of all-electric aircraft with operating costs and airspace complexity trade-off}


\author*[1]{\fnm{Lucas} \sur{Souza e Silva}}\email{lucas.souzaesilva@mail.concordia.ca}

\author[1]{\fnm{Luis} \sur{Rodrigues}}

\affil*[1]{\orgdiv{Electrical and Computer Engineering}, \orgname{Concordia University}, \city{Montreal}, \state{Quebec}, \country{Canada}}


\abstract{This paper formulates an optimal control framework for computing cruise airspeeds in predecessor–follower platoons of all-electric aircraft that balance operational cost and airspace complexity. To quantify controller workload and coordination effort, a novel pairwise dynamic workload (PDW) function is developed. 
Within this framework, the optimal airspeed solution is derived for all-electric aircraft under longitudinal wind disturbances. Moreover, an analytical suboptimal solution for heterogeneous platoons with nonlinear aircraft dynamics is determined, for which a general sufficient condition for string stability is formally established. The methodology is validated through case studies of all-electric aircraft operating in air corridors that are suitable for low-altitude advanced/urban air mobility (AAM/UAM) applications. Results show that the suboptimal solution closely approximates the optimal, while ensuring safe separations, maintaining string stability, and reducing operational cost and airspace complexity. These findings support the development of sustainable and more autonomous air traffic procedures that will enable the implementation of emerging air transportation technologies, such as AAM/UAM, and their integration to the air traffic system environment.}




\maketitle

\section{Introduction}\label{sec1}
The airspace is becoming more complex due to the growing demand for both domestic and international travel. As a result, the augmented airspace complexity imposes greater strain on the efficiency of air traffic management (ATM). This strain is  intense on the air traffic control operators (ATCOs), who must perform high-pressure, time-sensitive tasks that call for sustained attention over extended work shifts \cite{NYT}. Moreover, new transportation technologies, such as advanced/urban air mobility (AAM/UAM), require integration with the traditional air traffic control (ATC) environment. Thus, achieving a harmonized and coordinated air traffic management becomes increasingly challenging. One approach to alleviate these challenges is the use of structured and collaborative aircraft procedures. Several of these procedures can be characterized as platoons of pairwise predecessor–follower operations, in which safe separation between aircraft is maintained throughout the flight. All-electric aircraft operating within AAM/UAM air corridors serve as an example of these procedures, where coordinated cooperation will ensure safe and efficient traffic management, consistent with the guidelines outlined in \cite{FAA_Conops}. Furthermore, these procedures may support low-altitude aircraft operations in uncontrolled airspace sectors, where air traffic services are limited. 

As the complexity of the airspace continues to grow, these operations necessitate automated or supervised ATC strategies to ensure their effective management. However, disturbances may be amplified along automated predecessor-follower aircraft platoons, resulting in string instability. As a consequence, small speed adjustments by the predecessor may force large, oscillatory corrections by the followers, risking cascading separation violations and requiring more frequent ATC interventions. In addition to complying with all applicable airworthiness regulations and ensuring effective airspace coordination, aircraft operators must also remain profitable and meet customer expectations. This makes the consideration of operational costs essential to guarantee a sustainable airspace for all stakeholders. This paper proposes a novel optimal control framework to compute the airspeed of all-electric aircraft in longitudinal platoons of predecessor-follower procedures. To the best of the authors’ knowledge, this is the first framework to compute optimal cruise airspeeds for such procedures, aiming at balancing operational cost and airspace complexity. More specifically, the contributions of this paper are:

\begin{enumerate}
    \item The proposition of an optimal control framework for a class of pairwise predecessor-follower airspace procedures, applicable to all-electric aircraft operations, considering longitudinal wind.
    \item The proposal of a pairwise dynamic workload (PDW) function as an airspace complexity dynamic metric that models relevant airspace parameters of pairwise predecessor-follower aircraft operations.
    \item The derivation of a sufficient condition for string stability of pairwise all-electric aircraft operations, applicable to platoons modeled with nonlinear aircraft dynamics and heterogeneous parameter values. 
\end{enumerate}

The remainder of this paper is organized as follows. Section \ref{sec_RelatedW} reviews the relevant literature. Section \ref{sec_Prob_Stat_Sol} introduces the pairwise dynamic workload (PDW) metric for airspace complexity and presents the problem statement and formulation. Section \ref{sec_Sol} describes the proposed methodology and derives the optimal and suboptimal airspeed solutions for all-electric aircraft, along with a sufficient condition for string stability in pairwise operations. Section \ref{sec_Results} validates the assumption underlying the suboptimal solution using a shooting method and presents simulation results for realistic air-corridor scenarios, including the evaluation of longitudinal wind effects and string stability. Finally, Section \ref{sec_Conc} concludes the paper.

\section{Background and Related Work}\label{sec_RelatedW}
This paper considers aircraft operations in cruise of constant altitude $h_c$ characterized as longitudinal and heterogeneous platoons with unidirectional predecessor-follower topology subject to hybrid control, as shown in Fig. \ref{fig_platoon_gen}.

\textit{Definition 1 (Longitudinal Heterogeneous and Unidirectional Predecessor-Follower Platoon with Hybrid Control):} A platoon is said to be longitudinal, heterogeneous, and unidirectional with hybrid control if it satisfies the following properties: (i) it follows a \textit{unidirectional predecessor-follower topology}, where each follower agent receives information solely from its immediate predecessor, without feedback in the reverse direction; (ii) it is \textit{longitudinal}, meaning that only longitudinal motion is considered, with lateral and vertical dynamics neglected; (iii) it is \textit{heterogeneous}, comprising agents that are modeled with different physical and dynamic parameters; and (iv) it employs \textit{hybrid control}, wherein each agent applies a locally defined control law based on available information, while being coordinated or supervised by a centralized air traffic 

\begin{figure}[ht]
\centerline{\includegraphics[scale=0.5]{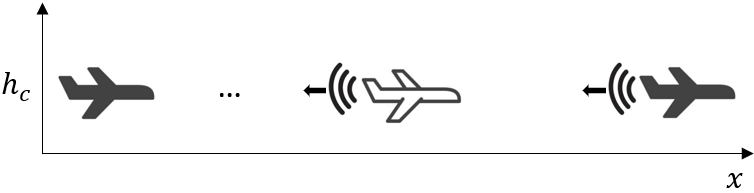}}
\caption{Longitudinal heterogeneous and unidirectional predecessor-follower aircraft platoon}
\label{fig_platoon_gen}
\end{figure}

An essential aspect of sustainable predecessor–follower operations within complex airspace sectors is the reduction of greenhouse gas (GHG) emissions. The economy mode \cite{SouzaeSilva2025} offers a complementary solution due to the limited energy density of modern batteries for all-electric (AEA) and hybrid aircraft \cite{Barzkar_Ghassemi2022}. In the economy mode, optimal parameters (e.g., airspeed) are computed to minimize operational costs, based on the ratio of time-related to energy-related costs, known as the cost index ($CI$).

\subsection {Predecessor-follower operations and formation flights} \label{subsec_PF_formations}

Predecessor–follower operations are a class of airspace procedures that maintain safe separation between aircraft. As airspace demand increases, such procedures offer automation opportunities to reduce complexity by improving coordination and limiting reliance on manual control. Examples include (i) AAM/UAM vehicles operating in air corridors, (ii) autonomous or supervised ATM operations, such as the Automatic Dependent Surveillance - Broadcast (ADS-B) In Flight-deck Interval Management (FIM) \cite{FAA_FIM}, (iii) cooperative ATM requiring real-time speed adjustments, (iv) separation in uncontrolled or low-surveillance airspace, (v) aircraft platooning, and (vi) terminal-area sequencing. 

Prior research on aircraft formations and predecessor-follower operations has largely focused on formation keeping \cite{Giulietti2000}, safe separation \cite{Brittain2019,Su2024}, and air traffic management applications, such as conflict resolution \cite{Tomlin1998}, airport operations \cite{Gorodetsky2008}, ATCO collaboration \cite{Wolfe2009}, arrival spacing \cite{Toratani2022}, and route planning \cite{Yang2018}. More recent work extends these methods to unmanned aerial vehicles (UAV) \cite{Chen2015,Su2023} and UAM platoons \cite{Mayle2023}; however, energy management for cost and GHG emissions reduction remains unaddressed. Early studies explored proportional-integral (PI) control for energy-efficient formations \cite{Buzogany1993} and constant-speed flight plans \cite{Ribichini2003,Benvenuti2024}, while nonlinear formation dynamics were examined without wind effects \cite{Hartjes2018,Curtis2025}. Subsequent work incorporated wind fields \cite{Hartjes2019,Cerezo2021,Tang2022, Cerezo2022}, and graph theory approaches to model UAM traffic with energy-efficient spacing \cite{Khoshdel2024}. As opposed to the previous literature, this paper develops optimal and suboptimal airspeed computation methods for pairwise predecessor-follower operations that explicitly balance operational costs and airspace complexity, applicable to all-electric aircraft under constant longitudinal winds.

\subsection{String stability in aircraft operations} \label{subsec_SS_aircraft}

String stability was formalized in \cite{Besselink2017} as the condition for which disturbances originating from preceding agents do not amplify as they propagate through the sequence of vehicles in a platoon.

\textit{Definition 2 (String stability):} An aircraft platoon with a predecessor-follower topology is string stable if any bounded perturbation in the predecessor’s airspeed $\delta v_{i-1}(t)$ yields a follower perturbation $\delta v_i(t)$ such that $|\delta v_i(t)| \leq |\delta v_{i-1}(t)|$. 

Platooning has been extensively investigated in the context of ground vehicles. A notable example is the California Partners for Advanced Transit and Highways (PATH) program, which represented the first large-scale, fully integrated, and systematically validated demonstration of automated vehicle platoons operating on highway-like environments \cite{Path}. A more recent example has extended the leader–follower framework to a specialized application of timber hauling on forest roads under challenging terrain conditions in Quebec, Canada \cite{Kratos}. In contrast to ground vehicle platoons, aircraft platooning introduces distinct operational constraints and aerodynamic considerations. Specifically, the aircraft airspeed is lower-bounded by the stall speed, preventing them from decelerating to a full stop to ensure safe separation, as is possible with ground vehicles. Moreover, the aerodynamic drag experienced by aircraft includes an induced component that arises from the lift generation and that is typically neglected in ground vehicle platoon formulations. References \cite{Turri2017} and \cite{Hussein2021} provide examples of ground vehicle platoons where the aerodynamic drag model accounts only for the parasitic (zero-lift) factor.

Sufficient conditions for string stability have been proposed across various analytical domains (e.g., time and frequency), for both linear and nonlinear dynamic systems, with the majority of studies focused on ground vehicle applications \cite{Feng2019}. However, only a limited number of articles have extended the concept of string stability to aircraft formations \cite{Allen2002,Swieringa2015,Weitz2018,Riehl2022,Stephens2023}. Early studies numerically analyzed linear formations \cite{Allen2002}, while sufficient conditions were formally derived for interval management \cite{Swieringa2015,Weitz2018}, though neglecting nonlinear dynamics and energy efficiency. More recent work addressed these gaps: \cite{Riehl2022} introduced energy-efficient control ensuring string stability, while \cite{Stephens2023} examined nonlinear heterogeneous formations under winds and turbulence. This paper establishes, for the first time, a sufficient condition for string stability in predecessor-follower operations of all-electric aircraft governed by a nonlinear function. Making this function equal to zero leads to the airspeed values that balance operational costs and airspace complexity for all-electric aircraft under longitudinal winds.

\subsection {Airspace complexity metrics} \label{subsec_Complexity_rev}
Several metrics have been proposed in the literature to quantify airspace complexity, depending on the type of operations of interest. For pairwise aircraft operating in cruise, it is essential to consider their horizontal separation and the rate at which this separation changes over time. References \cite{Pawlak1996, Chatterji2001, Laudeman1998} propose different complexity metrics and analyses that account for horizontal separation and separation rate. However, \cite{Pawlak1996} and \cite{Chatterji2001} do not provide a closed-form relation among these parameters, while \cite{Laudeman1998} only considers the number of aircraft in a sector that changed speed or are within unsafe distances from others, without capturing their continuous evolution over time. To address these limitations, we propose in section \ref{subsec_Complexity_Metric}
a novel pairwise dynamic workload (PDW) metric function for pairwise predecessor-follower aircraft operations.

\section{Problem statement and formulation}\label{sec_Prob_Stat_Sol}

\subsection {Assumptions and aircraft dynamic model} \label{subsec_Dynam_Assump}

Consider the number of aircraft in a platoon is $N+1$, with a follower aircraft labeled $A_i$ for $i=1,2,...,N$ and its predecessor aircraft labeled $A_{i-1}$. Safe predecessor-follower aircraft operations guarantee that the separation between $A_i$ and $A_{i-1}$ meets the applicable airspace management requirements, aiming at reduced ATC interference. Let $t$ be time, $x_i(t)$ the horizontal position of the follower aircraft in cruise, $v_i(t)$ its airspeed, $D_i(t)$, $L_i(t)$ and $T_i(t)$ the magnitude of the drag, lift and thrust forces applied to aircraft $A_i$, $W_i(t)$ its weight and let $v_{i-1}(t)$ be the airspeed of $A_{i-1}$. The following assumptions are made:

\begin{enumerate}
    \item Each follower aircraft cruises in steady flight, at constant altitude, which implies that the air density $\rho$ is also constant. As a consequence of this assumption we have $L_i(t)=W_i(t)$ and $T_i(t)=D_i(t)$.
    \item It is assumed that the Mach number for both aircraft remains below the drag divergence threshold, allowing to neglect the wave drag effects.
    \item Longitudinal winds of constant speed $v_w$ are considered, so that $v_w > 0$ represents a tailwind and $v_w < 0$ represents headwind.
    \item All aircraft in the platoon have fixed wings of surface area $S_i$, they operate within their flight envelope and have the same stall and maximum airspeeds.
    \item The electrical charge $Q_i(t)$ is provided by an ideal battery with negligible internal resistance, efficiency $\eta_i$ and voltage $U_i$ that is assumed constant during cruise.
    \item The follower aircraft is equipped with secondary surveillance systems (e.g., ADS-B In receivers) or sensors that allow it to access the predecessor's airspeed $v_{i-1}(t)$.
\end{enumerate}

Based on assumptions 1 and 2, the magnitude of the drag force $D_i(t)$ is given by \cite{Anderson1999}
\begin{equation}
    D_i(t) = \frac{1}{2} \rho S_i C_{{D,0}_i} v_i^2(t) + \frac{2C_{{D,2}_i}W_i^2}{\rho S_i v_i^2(t)} \label{eq_drag}
\end{equation}
where $C_{{D,0}_i}$ and $C_{{D,2}_i}$ are the parasitic zero-lift drag coefficient and the lift-induced drag coefficient, respectively. The weight of all-electric aircraft is constant throughout the flight, thus $W_i(t) = W_i$. The dynamic model of a follower aircraft cruising in predecessor-follower operations, accounting for the separation $d_{i,i-1}(t)$ from its predecessor, can be expressed as 
\begin{equation}
    \dot{x}_i(t) = v_i(t) + v_w \label{eq_dot_x}
\end{equation}
\begin{equation}
    \dot{d}_{i,i-1}(t) = v_{i-1}(t) - v_i(t)\label{eq_dot_d}
\end{equation}
\begin{equation}
  \dot{\Gamma}_i(t) = f_i(T_i(t),\Gamma_i(t)) = f_i(D_i(t),\Gamma_i(t))\label{eq_dot_Gamma}     
\end{equation}

In (\ref{eq_dot_Gamma}), $\dot{\Gamma}_i(t)$ represents the energy consumption rate of the follower aircraft. For all-electric aircraft, energy is provided by onboard batteries and $\dot{\Gamma}_i(t)$ is defined in Coulombs per second ($\mathrm{C/s}$) or Ampères ($\mathrm{A}$), corresponding to the electric current drawn from the batteries. As $T_i(t)=D_i(t)$ per assumption 1, then $f_i(T_i(t),\Gamma_i(t)) = f_i(D_i(t),\Gamma_i(t))$. 

\subsection{Pairwise dynamic workload (PDW)} \label{subsec_Complexity_Metric}
To dynamically describe the variation of the airspace complexity associated to pairwise predecessor-follower aircraft operations in cruise, it is assumed that both aircraft are flying at the same altitude. We propose a pairwise dynamic workload (PDW) function $\chi_{i,i-1}$ that satisfies the following conditions:

\begin{itemize}
    \item The function $\chi_{i,i-1}$ is continuous and strictly positive for all $t \in [0,t_f]$, where $t_f$ denotes the duration of the pairwise operation.
    
    \item The value of $\chi_{i,i-1}$ increases as the inter-aircraft separation $d_{i,i-1}$ decreases, and decreases when $d_{i,i-1}$ increases.
    
    \item An increase in the separation rate $\dot{d}_{i,i-1}$ should increase the derivative of $\chi_{i,i-1}$ with respect to $d_{i,i-1}$. Conversely, a decrease in $\dot{d}_{i,i-1}$ should reduce this derivative.
\end{itemize}

\textit{Definition 3 (Pairwise dynamic workload):} Given the separation $d_{i,i-1}$ and separation rate $\dot{d}_{i,i-1}$ of two aircraft in predecessor-follower operations of duration $t_f$ and given the maximum separation rate $\dot{d}_{max} > 0$, with $\dot{d}_{max} > |\dot{d}_{i,i-1}|$ $\forall$ $t \in [0,t_f]$, the proposed pairwise dynamic workload (PDW) function $\chi_{i,i-1}(d_{i,i-1},\dot{d}_{i,i-1})$ is
\begin{equation}
    \chi_{i,i-1}(d_{i,i-1},\dot{d}_{i,i-1}) = \frac{1}{d_{i,i-1}} \biggl(1-\frac{\dot{d}_{i,i-1}}{\dot{d}_{max}} \biggl) \label{eq_PDW}
\end{equation}

Note that $\chi_{i,i-1} > 0$ for any finite  $d_{i,i-1}$, and it is inversely proportional to the inter-aircraft separation, as smaller separations require increased attention and workload from the ATCOs. This behavior is confirmed by the derivative in \eqref{eq_partialX_d}, where the numerator is always positive, which makes the derivative of $\chi_{i,i-1}$ with respect to $d_{i,i-1}$ always negative.
\begin{equation}
    \frac{\partial \chi_{i,i-1}}{\partial d_{i,i-1}} = -\frac{1-\frac{\dot{d}_{i,i-1}}{\dot{d}_{max}}}{d_{i,i-1}^2} \label{eq_partialX_d}
\end{equation}

Moreover, a negative closure rate $\dot{d}_{i,i-1}$ indicates that the aircraft are converging, thereby increasing the coordination workload for the ATCOs and contributing to greater airspace complexity. In contrast, a positive closure rate implies that the aircraft are diverging, which reduces the controller's workload and consequently lowers the associated complexity. This interpretation can be confirmed by inspecting the computed derivative in \eqref{eq_partialX_ddot}. The derivative of $\chi_{i,i-1}$ with respect to $\dot{d}_{i,i-1}$ is always negative, given that the denominator in \eqref{eq_partialX_ddot} is always positive.
\begin{equation}
     \frac{\partial \chi_{i,i-1}}{\partial \dot{d}_{i,i-1}} = -\frac{1}{d_{i,i-1}\dot{d}_{max}} \label{eq_partialX_ddot}
\end{equation}

\subsection{Problem formulation} \label{subsec_Form}
In this paper, we propose a cost function that combines the operational cost of an aircraft engaged in predecessor-follower operations with the associated airspace complexity modeled as the PDW metric function introduced in section \ref{subsec_Complexity_Metric}. 

\textit{Definition 4 (Cost function):} Given the cost index $CI_i\ge0$, the energy consumption rate function $\dot{\Gamma}_i(t)$ of the follower aircraft and the pairwise dynamic workload $\chi_{i,i-1}$, the cost function $J_i$ is
\begin{equation}
    J_i =\int_{0}^{t_{f}} [CI_i - \dot{\Gamma}_i(t) + \alpha \chi_{i,i-1}]\,dt\label{eq_totalJ}
\end{equation}
where $\alpha\ge0$ is a scaling factor. The minimization of the cost in (\ref{eq_totalJ}) applied to a predecessor-follower operation can be formulated as an optimal control problem (OCP) as follows, 
\begin{equation}
\begin{aligned}
J_i^{*} = \min_{v_i,t_{f}} \quad & \int_{0}^{t_{f}} [CI_i - \dot{\Gamma}_i(t) + \alpha \chi_{i,i-1}] \,dt \label{eq_opt_prob}\\
\textrm{s.t.} \quad & \dot{x}_i(t)=v_i(t)+v_w\\
  &\dot{d}_{i,i-1}(t) = v_{i-1}(t) - v_i(t)\\
  &\dot{\Gamma}_i(t)=f_i(D_i(t),\Gamma_i(t))\\ 
  & \chi_{i,i-1} = \frac{1}{d_{i,i-1}(t)} \biggl(1-\frac{\dot{d}_{i,i-1}(t)}{\dot{d}_{max}} \biggl)\\
  &D_i(t) = \frac{1}{2} \rho S_i C_{{D,0}_i} v_i^2(t) + \frac{2C_{{D,2}_i}W_i^2}{\rho S_i v_i^2(t)}\\
  &\Gamma_i(0)=\Gamma_{{0}_i}, d_{i,i-1}(0)=d_{{0}_{i,i-1}}\\
  &x_i(0)=x_{{0}_i}, x_i(t_{f})=x_{{f}_i}\\
  &d_{i,i-1}(t) \geq d_{min} \\
  & v_{stall} < v_i(t) < v_{max}\\
\end{aligned}
\end{equation}
where $J_i^{*}$ is the minimum cost achieved for the minimizers of (\ref{eq_opt_prob}), which are the optimal follower aircraft airspeed $v_i^{*}(t)$ and optimal cruise time $t_{f}^{*}$. Note that the cost index $CI_i$ is a penalty factor in the final time $t_f^*$, so increasing $CI_i$ results in increasing $v_i^*(t)$. The minimum aircraft separation $d_{min}$, established by applicable air traffic regulations, ensures safe separation and avoids collisions. It is assumed that $d_{{0}_{i,i-1}}>d_{min}$. The value of $v_i(t)$ is lower bounded by the stall airspeed $v_{stall}$ of the follower aircraft and upper bounded by its maximum airspeed $v_{max}$.

\section{Problem solution} \label{sec_Sol}
This section derives the optimal and suboptimal airspeed solutions for all-electric aircraft platoons and establishes a sufficient condition for their string stability.

\subsection{OCP solution} \label{subsec_opt_Sol}

\begin{theorem} [Optimal cruise airspeed]
    Let $I_t$ be an interval of positive length such that $I_t=(t_1,t_2)\subset[0,t^*_f]$, with $t_2>t_1>0$ and $t_s \in [0,t^*_f] $ denotes a time instant. We define the time-dependent multiplier $\mu(t) = \mu_{b}(t) + \mu_s \delta(t-t_s)$, with $\mu_{b}(t) \geq 0$, $\mu_s \geq 0$ and $\delta(t-t_s)$ is the Dirac delta function. The optimal airspeed $v_i^*(t)$ of a follower all-electric aircraft in a longitudinal platoon satisfies the following case-dependent conditions:

    \begin{enumerate} [wide]
 \item \textbf{Interior region ($d_{i,i-1}(t) > d_{min}$):} The optimal airspeed $v_i^*(t)$ is the solution of $p_i=0$, with the function $p_i$ defined in \eqref{eq_OCP_sol}, if $v_{stall} < v_i^*(t) < v_{max}$.
 \small
    \begin{align}
      & p_i = CI_i + J_{d_{i,i-1}}^*(t) (v_{i-1}(t)+v_w) + \frac{\alpha}{d_{i,i-1}(t)} \Biggl[1-\frac{v_{i-1}(t)+v_w}{\dot{d}_{max}} \Biggl] \nonumber \\
      &+ \frac{2C_{{D,2}_i}W_i^2}{\eta_i U_i\rho S_i {v_i^*}^2(t)}(2v_i^*(t) + v_w)
      -\frac{\rho S_i C_{{D,0}_i} {v_i^*}^2(t)}{\eta_i U_i} \biggl(v_i^*(t) + \frac{3}{2}v_w\biggl) 
      & \label{eq_OCP_sol}
      \end{align}
      \normalsize
with costate equation:
\begin{equation}
    \dot{J}_{d_{i,i-1}}^*(t) = \frac{\alpha}{d^2_{i,i-1}(t)} \biggl(1 + \frac{v_i^*(t)-v_{i-1}(t)}{\dot{d}_{max}}\biggl)\label{eq_costate_Jd}
\end{equation}
and final condition:
\begin{equation}
    J_{d_{i,i-1}}^*(t_f^*)=0 
\end{equation}

    \item \textbf{Instantaneous constraint boundary contact at $t = t_s$ ($d_{i,i-1}(t_s)=d_{min}$ and $\dot{d}_{i,i-1}(t_s)>0$):} the optimal airspeed $v_i^*(t_s)$ is the solution of $q_i=0$, with the function $q_i$ given by \eqref{eq_q_i} and $t_s^- = \lim_{\epsilon \rightarrow0^+}(t_s-\epsilon)$, if  $v_{stall} < v_i^*(t_s) < v_{max}$.
    \begin{align}
      & q_i =  [v_i^*(t_s^-)-v_i^*(t_s)]\Biggl[ -\frac{3\rho S_i C_{{D,0}_i}}{2\eta_i U_i}{v_i^*}^2(t_s^-) + \frac{2C_{{D,2}_i}W_i^2}{\eta_i U_i\rho S_i v_i^*(t_s^-)}\Biggl] \nonumber \\
      & -\frac{\rho S_i C_{{D,0}_i}}{\eta_i U_i} [{v_i^*}^3(t_s)-{v_i^*}^3(t_s^-)] - \frac{2C_{{D,2}_i}W_i^2}{\eta_i U_i\rho S_i} \Biggl[\frac{1}{v_i^*(t_s)} - \frac{1}{v_i^*(t_s^-)} \Biggl] \nonumber \\  
       & +\frac{\alpha}{d_{i,i-1}(t_s^-)}[v_{i-1}(t_s) - v_{i-1}(t_s^-)] \nonumber \\
    &- J_{d_{i,i-1}}^*(t_s)[v_{i-1}(t_s)-v_i^*(t_s)] +J_{d_{i,i-1}}^*(t_s^-)[v_{i-1}(t_s^-)-v_i^*(t_s)]\label{eq_q_i}
      \end{align}
with costate equation
\begin{equation}
    \dot{J}_{d_{i,i-1}}^*(t_s) = \frac{\alpha}{d^2_{i,i-1}(t_s)} \biggl(1 + \frac{v_i^*(t_s)-v_{i-1}(t_s)}{\dot{d}_{max}}\biggl) + \mu_s\label{eq_costate_Jd_contact}
\end{equation}
and additional conditions
\begin{equation}
\mu_s>0, \ \mu_s(d_{min} - d_{i,i-1}(t_s)) =0
\end{equation}
\begin{equation}
 J_{d_{i,i-1}}^*(t_s)=J_{d_{i,i-1}}^*(t_s^-) +\mu_s
\end{equation}

    \item \textbf{Constraint boundary arc, for which $t \in I_t$ ($d_{i,i-1}(t)=d_{min}$ and $\dot{d}_{i,i-1}(t)=0$):} the optimal airspeed $v_i^*(t)$ is
    \begin{equation}
        v_i^*(t) = v_{i-1}(t)
    \end{equation}
with costate equation
\begin{equation}
    \dot{J}_{d_{i,i-1}}^*(t) = \frac{\alpha}{d^2_{i,i-1}(t)} \biggl(1 + \frac{v_i^*(t)-v_{i-1}(t)}{\dot{d}_{max}}\biggl) + \mu_b(t)\label{eq_costate_Jd_bound}
\end{equation}
entry conditions for $t_1^- = \lim_{\epsilon \rightarrow0^+}(t_1-\epsilon)$
\\
\\
    \begin{align}
      & [v_i^*(t_1^-)-v_i^*(t_1)]\Biggl[ -\frac{3\rho S_i C_{{D,0}_i}}{2\eta_i U_i}{v_i^*}^2(t_1^-) + \frac{2C_{{D,2}_i}W_i^2}{\eta_i U_i\rho S_i v_i^*(t_1^-)}\Biggl] \nonumber \\
      & -\frac{\rho S_i C_{{D,0}_i}}{\eta_i U_i} [{v_i^*}^3(t_1)-{v_i^*}^3(t_1^-)] - \frac{2C_{{D,2}_i}W_i^2}{\eta_i U_i\rho S_i} \Biggl[\frac{1}{v_i^*(t_1)} - \frac{1}{v_i^*(t_1^-)} \Biggl] \nonumber \\  
       & +\frac{\alpha}{d_{i,i-1}(t_1^-)}[v_{i-1}(t_1) - v_{i-1}(t_1^-)]+J_{d_{i,i-1}}^*(t_1^-)[v_{i-1}(t_1^-)-v_i^*(t_1)] \nonumber \\
    &- J_{d_{i,i-1}}^*(t_1)[v_{i-1}(t_1)-v_i^*(t_1)] = 0\label{eq_entry_cond}
      \end{align}
and 
\begin{equation}
J_{d_{i,i-1}}^*(t_1)=J_{d_{i,i-1}}^*(t_1^-) +\mu_b(t_1)
\end{equation}
with 
\begin{equation}
\mu_b(t)>0, \ \mu_b(t)(d_{min} - d_{i,i-1}(t)) =0
\end{equation}
   
    \end{enumerate}
\end{theorem}

\textit{Proof:} The Hamiltonian is given by
\begin{align}
    & H = CI_i + (J_{{\Gamma}_i}^*(t)-1)f_i(D_i(t),\Gamma_i(t)) +\frac{\alpha}{d_{i,i-1}(t)} \biggl(1-\frac{v_{i-1}(t)-v_i(t)}{\dot{d}_{max}}\biggl)  \nonumber \\
    & + J_{{x}_i}^*(t) (v_i(t)+v_w) + J_{d_{i,i-1}}^*(t) (v_{i-1}(t) - v_i(t)) +\mu(t)(d_{min}-d_{i,i-1}(t))  \label{eq_Hamiltonian_1}
\end{align}
where $J_{{x}_i}^*(t) $, $J_{{\Gamma}_i}^*(t)$ and $J_{d_{i,i-1}}^*(t)$ represent the partial derivatives of the optimal cost $J_i^*$ with respect to $x_i(t)$, $\Gamma _i(t)$ and $d_{i,i-1}(t)$, respectively. The parameter $\mu(t) = \mu_{b}(t) + \mu_s \delta(t-t_s)$ is a multiplier that satisfies the Karush–Kuhn–Tucker (KKT) conditions for all time $t\in [0,t_f^*]$:
\begin{equation}
    \mu(t)\geq0, \ \mu(t)(d_{min} - d_{i,i-1}(t)) =0
\end{equation}

The terminal condition $\Psi_i$ can be determined as 
\begin{equation}
    \Psi_i = x_i(t_f^*) - x_{f_i} = 0 \label{eq_terminal_cond}
\end{equation}

Based on (\ref{eq_Hamiltonian_1}) and (\ref{eq_terminal_cond}), the transversality conditions \cite{Pontr} are
\begin{equation}
    H(t^*_f) = 0 \label{eq_transv_H}
\end{equation}
\begin{equation}
    J_{{\Gamma}_i}^*(t_f^*) = \nu_\Gamma \frac{\partial \Psi_i}{\partial \Gamma_i(t)} \biggl|_{t^*_{f}} = 0 \label{eq_transv_Jw}
\end{equation}
\begin{equation}
    J_{d_{i,i-1}}^*(t_f^*) = \nu_d \frac{\partial \Psi_i}{\partial d_{i,i-1}(t)} \biggl|_{t^*_{f}} = 0 \label{eq_transv_Jd}
\end{equation}
where $\nu_\Gamma$ and $\nu_d$ are Lagrange multipliers. 

For all-electric aircraft under the assumptions 1 to 5 we have $\dot{\Gamma}_i(t) = f_i(D_i(t),\Gamma_i(t)) = \dot{Q}_i(t) = -\frac{D_i(t)v_i(t)}{\eta_iU_i}$, which results in
\begin{equation}
    f_i(D_i(t),\Gamma_i(t)) = -\frac{1}{\eta_i U_i }\Biggl(\frac{1}{2} \rho S_i C_{{D,0}_i} v_i^3(t) + \frac{2C_{{D,2}_i}W_i^2}{\rho S_i v_i(t)} \Biggl)\label{eq_f_elec}
\end{equation}
\begin{equation}
    \frac{\partial f_i(D_i(t),\Gamma_i(t))}{\partial v_i(t)} = -\frac{1}{\eta_i U_i }\Biggl(\frac{3}{2} \rho S_i C_{{D,0}_i} v_i^2(t) - \frac{2C_{{D,2}_i}W_i^2}{\rho S_i v_i^2(t)} \Biggl)\label{eq_df_elec}
\end{equation}

As \eqref{eq_f_elec} and \eqref{eq_df_elec} do not depend on the electrical charge $Q_i(t)$, the costate equation for $J_{{\Gamma}_i}^*(t)$ becomes
\begin{equation}
    \dot{J}_{{\Gamma}_i}^*(t) = - \frac{\partial H}{\partial \Gamma_i(t)} = - \frac{\partial H}{\partial Q_i(t)} = 0 \label{eq_cost_J_gamma}
\end{equation}

With the results from \eqref{eq_transv_Jw} and \eqref{eq_cost_J_gamma}, we can conclude that $J_{{\Gamma}_i}^*(t) =0 \ \forall t \in [0,t_f^*]$. The dynamics of the remaining costates are given by
\begin{equation}
    \dot{J}_{{x}_i}^*(t) = - \frac{\partial H}{\partial x_i(t)} = 0 \label{eq_costate_Jx}
\end{equation}
\begin{equation}
    \dot{J}_{d_{i,i-1}}^*(t) = - \frac{\partial H}{\partial d_{i,i-1}(t)} = \frac{\alpha}{d_{i,i-1}^2(t)} \biggl(1 + \frac{v_i^*(t)-v_{i-1}(t)}{\dot{d}_{max}}\biggl) + \mu(t)\label{eq_costate_Jd_gen}
\end{equation}

The solution of (\ref{eq_opt_prob}) requires the minimization of the Hamiltonian described in (\ref{eq_Hamiltonian_1}) with respect to the follower aircraft airspeed $v_i(t)$. The first-order necessary condition for a minimum must be satisfied as per Pontryagin's Minimum Principle \cite{Pontr}
\begin{align}
    & \frac{\partial H}{\partial v_i(t)} = \frac{3\rho S_i C_{{D,0}_i} {v_i^*}^2(t)}{2\eta_i U_i}  - \frac{2C_{{D,2}_i}W_i^2}{\eta_i U_i\rho S_i {v_i^*}^2(t)} + \frac{\alpha}{d_{i,i-1}(t)\dot{d}_{max}}+ J_{x_i}^*(t) - J_{d_{i,i-1}}^*(t) = 0 \label{eq_nec_cond}
\end{align}

The Hamiltonian is not explicitly dependent on time, thus we can establish that $\dot{H} = \partial _tH = 0$, which from (\ref{eq_transv_H}) yields
\begin{equation}
     H = H(t^*_f) = 0 \label{eq_H_zero}
\end{equation}

The remainder of the proof is structured into three cases, summarized in Table \ref{tab_cand}, each corresponding to a candidate optimal solution that depends on the values of the aircraft separation $d_{i,i-1}(t)$, the separation rate $\dot{d}_{i,i-1}(t)$ and the multiplier $\mu(t)$.

\begin{table} [!b]
\centering
   \caption{Optimal solution candidates}\label{tab_cand}
   \centering
   \begin{tabular}{|c |c |c |c |}
   \hline
     \textbf{Case} & \textbf{Aircraft separation} & \textbf{Separation rate} & \textbf{$\mu(t)$}  \\
     \hline
         1 & $d_{i,i-1}(t) > d_{min}$  & - & 0  \\
         2 & $d_{i,i-1}(t_s) = d_{min}$ & $\dot{d}_{i,i-1}(t_s) > 0$ & $\mu_s>0$ \\
         3 & $d_{i,i-1}(t) = d_{min}$ & $\dot{d}_{i,i-1}(t) = 0$ & $\mu_{b}(t)>0$ \\
     \hline
     \end{tabular}
\end{table}

\textbf{Case 1:} For $d_{i,i-1}(t) > d_{min}$, regardless of the value of $\dot{d}_{i,i-1}(t)$, the inequality constraint relative to the aircraft separation is inactive, which leads to $\mu(t) = 0$ in (\ref{eq_costate_Jd_gen}) yielding \eqref{eq_costate_Jd}. Then, solving (\ref{eq_nec_cond}) for $J_{x_i}^*(t)$ results in
\begin{align}
    & J_{x_i}^*(t) = J_{d_{i,i-1}}^*(t)  -\frac{3\rho S_i C_{{D,0}_i} {v_i^*}^2(t)}{2\eta_i U_i}  + \frac{2C_{{D,2}_i}W_i^2}{\eta_i U_i\rho S_i {v_i^*}^2(t)} - \frac{\alpha}{d_{i,i-1}(t)\dot{d}_{max}} \label{eq_Jx_H}
\end{align}

Replacing \eqref{eq_f_elec} and  (\ref{eq_Jx_H}) in (\ref{eq_Hamiltonian_1}) and using the result from (\ref{eq_H_zero}), one obtains (\ref{eq_OCP_sol}). 

\textbf{Case 2:} Consider the time instant $t=t_s$ at which $d_{i,i-1}(t_s) = d_{min}$. To ensure the separation between the follower aircraft and its predecessor remains larger than the safe boundary, the value of the separation rate becomes $\dot{d}_{i,i-1}(t_s)>0$. As a consequence, the multiplier $\mu(t_s) = \mu_s$ produces a jump in the costate $J_{d_{i,i-1}}^*(t_s)$ \cite{Hartl1995}, yielding
\begin{equation}
    J_{d_{i,i-1}}^*(t_s)=J_{d_{i,i-1}}^*(t_s^-) +\mu_s
\end{equation}

Moreover, as the inequality constraint $d_{i,i-1}(t_s) \geq d_{min}$ does not depend explicitly on time, the Hamiltonian is continuous, yielding
\begin{equation}
    H(t_s)=H(t_s^-) \label{eq_Hamil_contin}
\end{equation}
where $H(t_s^-)$ is the value of the Hamiltonian immediately before the constraint becomes active computed at $t = t_s^- = \lim_{\epsilon \rightarrow0^+}(t_s-\epsilon)$. The costate $J_{x_i}^*(t)$ does not depend on the multiplier $\mu(t)$ and based on \eqref{eq_costate_Jx} we can establish that
\begin{equation}
    J_{x_i}^*(t_s) = J_{x_i}^*(t_s^-) \label{eq_Hx_contin}
\end{equation}

Using the conditions from \eqref{eq_f_elec}, \eqref{eq_Hamil_contin} and \eqref{eq_Hx_contin} in \eqref{eq_Hamiltonian_1}, yields \eqref{eq_q_i}. The costate dynamics in \eqref{eq_costate_Jd_contact} is computed by setting $t = t_s$ and $\mu(t) = \mu_s$ in \eqref{eq_costate_Jd_gen}

\textbf{Case 3:} Consider the time $t \in I_t$ at which the aircraft separation remains equal to the minimum value, resulting in $d_{i,i-1}(t) = d_{min}$ and $\dot{d}_{i,i-1}(t)=0$. Thus, during the interval $I_t$, the optimal airspeed $v_i^*(t)$ is
\begin{equation}
    v_i^*(t) = v_{i-1}(t)
\end{equation}

The value of the inequality multiplier becomes $\mu(t) = \mu_b(t)$. As a consequence, the dynamics for the costate $J_{d_{i,i-1}}^*(t)$ are computed by setting $\mu(t) = \mu_b(t)$ in \eqref{eq_costate_Jd_gen}, resulting in \eqref{eq_costate_Jd_bound}. We define $t_1 \in I_t$ as the entry time instant at which the state $d_{i,i-1}(t_1) = d_{min}$. Similarly to Case 2, the Hamiltonian and the costate $J_{x_i}^*(t)$ are continuous, then for $t = t_1^- = \lim_{\epsilon \rightarrow0^+}(t_1-\epsilon)$, we compute $H(t_1)=H(t_1^-)$ and $J_{x_i}^*(t_1) = J_{x_i}^*(t_1^-)$. Using these conditions in \eqref{eq_Hamiltonian_1}, it results in the entry condition \eqref{eq_entry_cond}.

Furthermore, the resulting optimal candidate solutions for $v_i^*(t) \ \forall t \in [0,t_f^*]$ are feasible if $v_{stall} < v_i^*(t) < v_{max}$. Thus, since
\begin{equation}
    \frac{\partial^2 H}{\partial v_i^2(t)} = \frac{3 (\rho S_i)^2C_{{D,0}_i}{v_i}^4(t) + 4C_{{D,2}_i}W_i^2}{\eta_i U_i\rho S_i {v_i}^3(t)} > 0
\end{equation}
the Hamiltonian is strictly convex in $v_i(t)$, therefore, the feasible candidate solutions for $v_i^*(t)$ are minimizers of $H$. This concludes the proof. \qed

The optimal solution derived in Theorem 1 requires the evaluation of the costate $J_{d_{i,i-1}}^*(t)$ dynamics. However, this costate satisfies a zero terminal condition and exhibits only small variations along the optimal trajectory. This observation motivates the following approximation, which leads to a computationally more tractable suboptimal solution.

\begin{corollary}[Suboptimal airspeed]
        Let $v_{\varepsilon,i}(t) \geq 0$ be an airspeed adjust term. The suboptimal airspeed $v_i(t)$ of a follower all-electric aircraft engaged in a longitudinal predecessor-follower platoon as per the OCP (\ref{eq_opt_prob}) is given by \eqref{eq_suboptimal_cases}, if $v_{stall} < v_i(t) < v_{max}$.

\begin{equation}
v_i(t) =
  \begin{cases}
    v_{i-1}(t) - v_{\varepsilon,i}(t), & \text{if } d_{i,i-1}(t) = d_{min} \\[6pt]
    v_i(t) \in \mathbb{R}_{>0} \
    \text{such that } g_i=0, & \text{if } d_{i,i-1}(t) > d_{min} \\[6pt]
  \end{cases} \label{eq_suboptimal_cases}
\end{equation}
where the function $g_i$ is
\begin{align}
   & g_i = CI_i  + \frac{\alpha}{d_{i,i-1}(t)} \Biggl[1-\frac{v_{i-1}(t)+v_w}{\dot{d}_{max}} \Biggl]\nonumber \\
      &   -\frac{\rho S_i C_{{D,0}_i} {v_i}^2(t)}{\eta_i U_i} \biggl(v_i(t) + \frac{3}{2}v_w\biggl)+ \frac{2C_{{D,2}_i}W_i^2}{\eta_i U_i\rho S_i {v_i}^2(t)}(2v_i(t) + v_w)  \label{eq_OCP_subop}
\end{align}
and the airspeed adjustment term $v_{\varepsilon,i}(t)$ is computed by
\begin{equation}
v_{\varepsilon,i}(t) =
  \begin{cases}
     v_{\varepsilon,i}(t) \in \mathbb{R}_{>0} \
    \text{such that } y_i=0, & \text{if } \dot{d}_{i,i-1}(t) > 0 \\[6pt]
    0, & \text{if } \dot{d}_{i,i-1}(t) = 0 \\[6pt]
  \end{cases} \label{eq_v_eps_sub_cases}
\end{equation}
with the function $y_i$ defined as
\begin{align}
   &  y_i = CI_i  + \frac{\alpha}{d_{i,i-1}(t)} \Biggl[1-\frac{v_{i-1}(t)+v_w}{\dot{d}_{max}} \Biggl]\nonumber \\
      & -\frac{\rho S_i C_{{D,0}_i} ( v_{i-1}(t) - v_{\varepsilon,i}(t))^2}{\eta_i U_i} \biggl(v_{i-1}(t) - v_{\varepsilon,i}(t) + \frac{3}{2}v_w\biggl) \nonumber \\
      & + \frac{2C_{{D,2}_i}W_i^2}{\eta_i U_i\rho S_i ( v_{i-1}(t) - v_{\varepsilon,i}(t))^2}[2(v_{i-1}(t) - v_{\varepsilon,i}(t)) + v_w]   \label{eq_OCP_subop_v_eps}
\end{align}
\end{corollary}

\textit{Proof:} Assuming that the time variation of $J_{d_{i,i-1}}^*(t)$ is sufficiently small, we approximate $J_{d_{i,i-1}}^*(t) = J_{d_{i,i-1}}^*(t^*_f) = 0$ in (\ref{eq_OCP_sol}) resulting in (\ref{eq_OCP_subop}). The result in (\ref{eq_OCP_subop_v_eps}) is obtained by setting $v_i(t) = v_{i-1}(t) - v_{\varepsilon,i}(t)$ in (\ref{eq_OCP_subop}). \qed

\textbf{Remark:} The approximation that resulted in the suboptimal solution will be validated by numerically solving the proposed OCP using a shooting method as shown in section \ref{subsec_Shoot}.

\algrenewcommand\alglinenumber[1]{\scriptsize #1}

\newcommand{\AlgBlock}[2]{%
  \State \textbf{#1:} #2
}
\newcommand{\AlgSubBlock}[2]{%
  \State \quad \textbf{#1:} #2
}

\subsection{String stability of pairwise aircraft operations} \label{subsec_strstb}
This section derives a sufficient condition for string stability
in all-electric predecessor–follower longitudinal platoons that satisfy Definition 2.

\begin{theorem}[String stability]
Consider a unidirectional, predecessor-follower, longitudinal platoon of all-electric aircraft governed by the control law $g_i = 0$, with $g_i$  defined in (\ref{eq_OCP_subop}), where $v_i(t)$ and $v_{i-1}(t)$ denote the airspeeds of the follower and predecessor agents, respectively, and $g_i$ is continuously differentiable with respect to $v_i(t)$ and $v_{i-1}(t)$, with $\frac{\partial g_i}{\partial v_i(t)} \neq 0$ and $\frac{\partial g_i}{\partial v_{i-1}(t)} \neq 0$.  Let $\bar{v}_{i-1}(t) = v_{i-1}(t) +\delta v_{i-1}(t)$ be the predecessor's airspeed disturbed by a bounded quantity $\delta v_{i-1}(t)$ and the set $I_{v_{i-1}}$ be defined as the interval $I_{v_{i-1}}=(v_{i-1}(t), v_{i-1}(t) + \delta v_{i-1}(t))$. The platoon is string stable if the string stability coefficient $K_{i,i-1}(t) \leq 1 \ \forall t \in [0,t_f] $, where $K_{i,i-1}(t)$ is
\begin{align}
    & K_{i,i-1}(t) = \Biggl|-\frac{\eta_i U_i \alpha \rho S_i v_i^3(t)}{d_{i,i-1}(t)\dot{d}_{max}(v_i(t) + v_w)[3(\rho S_i)^2C_{{D,0}_i}v_i^4(t)+4C_{{D,2}_i}W_i^2]} \Biggl| \label{eq_SS_cond_applied}
\end{align}
\end{theorem}

\textit{Proof:} From $g_i = 0$, it follows, via the Implicit Function Theorem, that there exists a differentiable function $v_i(t) = \phi(v_{i-1}(t))$, assuming $\frac{\partial g_i}{\partial v_i(t)} \neq 0$ and $\frac{\partial g_i}{\partial v_{i-1}(t)} \neq 0$. Consider a bounded disturbance $\delta v_{i-1}(t)$ in the predecessor's airspeed, such that $|\delta v_{i-1}(t)|<\bar{\epsilon}$ for any $\bar{\epsilon}>0$. The disturbed predecessor's airspeed is
\begin{equation}
    \bar{v}_{i-1}(t) = v_{i-1}(t) +\delta v_{i-1}(t),
\end{equation}
which induces a corresponding disturbance in the follower’s airspeed given by
\begin{equation}
    \bar{v}_{i}(t) = \phi(\bar{v}_{i-1}(t)) = \phi(v_{i-1}(t) +\delta v_{i-1}(t))
\end{equation}

Hence, the resulting deviation in the follower’s airspeed is
\begin{align}
    \delta v_i(t) & = \phi(\bar{v}_{i-1}(t))-\phi(v_{i-1}(t))  = \phi(v_{i-1}(t) +\delta v_{i-1}(t)) - \phi(v_{i-1}(t)) \label{eq_follower_dev}
\end{align}

The Mean Value Theorem (MVT) states that $\phi(b) - \phi(a) = \phi^\prime (c)(b-a)$ for $\phi$ continuous in $[a,b]$ and differentiable in $(a,b)$, with $c \in (a,b)$. Applying the MVT to (\ref{eq_follower_dev}) with $a = v_{i-1}(t)$ and $b = v_{i-1}(t) + \delta v_{i-1}(t)$ yields
\begin{equation}
    \phi(v_{i-1}(t)+ \delta v_{i-1}(t)) - \phi(v_{i-1}(t)) = \phi^ \prime (c) \delta v_{i-1}(t), c \in I_{v_{i-1}} \label{eq_MVT_dev}
\end{equation}

From the result in (\ref{eq_MVT_dev}), we can establish that 
\begin{align}
    |\delta v_i(t)| & = |\phi(v_{i-1}(t) +\delta v_{i-1}(t)) - \phi(v_{i-1}(t))| = |\phi^ \prime (c)| |\delta v_{i-1}(t)| \label{eq_MVT_expand}
\end{align}

As $|\phi^ \prime (c)| \leq \sup_{v_{i-1}(t) \in I_{v_{i-1}}}|\phi^ \prime (v_{i-1}(t))|$ we can also write 
\begin{equation}
    |\delta v_i(t)| \leq \sup_{v_{i-1}(t) \in I_{v_{i-1}}}|\phi^ \prime (v_{i-1}(t))| |\delta v_{i-1}(t)| \label{eq_SS_cond_sup}
\end{equation}

To satisfy the string stability condition as specified in \textit{Definition 2}, it suffices that
\begin{equation}
   \sup_{v_{i-1}(t) \in I_{v_{i-1}}}|\phi^ \prime (v_{i-1}(t))| \leq 1 \label{eq_suf_cond_SS}
\end{equation}

From the Implicit Function Theorem, the derivative of the implicit relation is given by
\begin{equation}
    |\phi^ \prime (v_{i-1}(t))| = \Biggl|\frac{\partial v_i(t)}{\partial v_{i-1}(t)}\Biggl| = \Biggl|- \frac{\frac{\partial g_i}{\partial v_{i-1}(t)}}{\frac{\partial g_i}{\partial v_i(t)}}\Biggl| \label{eq_IFT_res}
\end{equation}

Therefore, combining (\ref{eq_suf_cond_SS}) and (\ref{eq_IFT_res}), the general sufficient condition for string stability is established as 
\begin{equation}
    K_{i,i-1}(t) = \Biggl|- \frac{\frac{\partial g_i}{\partial v_{i-1}(t)}}{\frac{\partial g_i}{\partial v_i(t)}}\Biggl| \quad \leq \quad 1. \label{eq_SS_cond}
\end{equation}

Considering the function $g_i$ defined in (\ref{eq_OCP_subop}), the following partial derivatives are computed
\begin{equation}
    \frac{\partial g_i}{\partial v_{i-1}(t)} = -\frac{\alpha}{d_{i,i-1}(t)\dot{d}_{max}} \label{eq_partialg_partialvi1}
\end{equation}
\begin{equation}
    \frac{\partial g_i}{\partial v_{i}(t)} = -\frac{(v_i(t) + v_w)}{\eta_i U_i} \Biggl(3 \rho S_i C_{{D,0}_i}v_i^2(t) + \frac{4 C_{{D,2}_i} W_i^2}{\rho S_i v_i^2(t)} \Biggl)\label{eq_partialg_partialvi}
\end{equation}

Using (\ref{eq_partialg_partialvi1}) and (\ref{eq_partialg_partialvi}) in (\ref{eq_SS_cond}), results in (\ref{eq_SS_cond_applied}). \qed

\section{Results and validation}\label{sec_Results}

\subsection{Simulation parameters} \label{subsec_Sim_Par}
Simulations were conducted in MATLAB on a laptop with a configuration of 16 GB RAM and an 11$^{th}$ Gen Intel$^R$ Core$^{TM}$ i5-1135G7 2.40GHz CPU. Two aircraft models with parameters defined in Table \ref{tab_sim_param} are considered, as follows:

\begin{itemize}
    \item Model $A_1$: Yuneec International E430 \cite{E430}.
    \item Model $A_2$: Pipistrel Velis Electro \cite{Velis}
\end{itemize}

\begin{table}[htbp]
\centering
\caption{Simulation Parameters}
\label{tab_sim_param}

\begin{tabular}{|c c c|}
\hline
\textbf{Parameter} & \textbf{$A_1$} & \textbf{$A_2$} \\
\hline
$S_i$ $[\mathrm{m^{2}}]$  & 11.37 & 9.5* \\ 
$W_i$ $[\mathrm{kN}]$ & 4.61\textsuperscript{1} & 5.40\textsuperscript{1} \\
$C_{{D,0}_i}$ & 0.035\textsuperscript{2} & 0.039* \\
$C_{{D,2}_i}$ & 0.009\textsuperscript{2} & 0.008* \\
$v_{i,max}$ $[\mathrm{km/h}]$ & 150 & 172 \\
$U_i$ $[\mathrm{V}]$ & 133.2 & 200* \\
$\eta_i$ & 0.7\textsuperscript{2} & 0.7* \\
$Q_i(0)$ $[\mathrm{C}]$ & $3.6\times10^{5}$* & $3\times10^{5}$* \\[1ex]
\hline
\end{tabular}

\begin{tablenotes}
\item[*] Estimated.
\item[1] Aircraft maximum take-off weight, which is constant throughout the flight.
\item[2] From \cite{LiRodrigues2023}.
\end{tablenotes}

\end{table}

\subsection{Pairwise dynamic workload metric} \label{subsec_compl_metric_eval}

The novel pairwise dynamic workload (PDW) function defined in (\ref{eq_PDW}) is further analyzed in this section. A comparison of the proposed PDW with other airspace complexity metrics from the literature is shown in Table \ref{table_matric_comparison}.

\begin{table} [hbt]
   \caption{Comparison of Airspace Complexity Metrics for Pairwise Aircraft Operations}\label{table_matric_comparison}
   \centering
   \begin{tabular}{|c| c c c c|}
   \hline
   \textbf{Parameter} & \textbf{\cite{Pawlak1996}} & \textbf{\cite{Chatterji2001}} & \textbf{\cite{Laudeman1998}} & \textbf{PDW} \\
   \hline
   Horizontal Separation & \checkmark & \checkmark & \checkmark & \checkmark \\
   Separation Rate & \checkmark & \checkmark & \checkmark & \checkmark \\
   Closed-form & \ding{55} & \ding{55} & \checkmark & \checkmark \\
   Temporal evolution of parameters & \ding{55} & \checkmark & \ding{55} & \checkmark \\
     \hline
     \end{tabular}
\end{table}

The metrics presented in \cite{Pawlak1996} and \cite{Chatterji2001} do not have a closed-form equation that relates the aircraft separation and separation rate. For this reason, the proposed PDW metric will be evaluated alongside the dynamic density (DD) measure introduced in \cite{Laudeman1998}. The comparison is conducted for two distinct scenarios involving a pairwise predecessor–follower procedure over a 60-minute duration. In these scenarios, $v_0(t)$ and $v_1(t)$ denote the airspeed of the predecessor and the follower aircraft, respectively.  The profiles and values of $v_0(t)$ and $v_1(t)$ were selected to generate multiple airspeed variations while maintaining a constant sign of $\dot{d}_{i,i-1}(t)$. The two scenarios are described below.

\begin{enumerate}
    \item Scenario 1, aircraft converging: the predecessor and follower aircraft airspeeds are, respectively,

    \begin{equation}
       v_0(t) =  \begin{cases}
           600 \ \mathrm{km/h}, \text{for} \ t \in [0,10), [20,30), [40,50) \\
           650 \ \mathrm{km/h}, \text{for} \ t \in [10,20), [30,40), [50,60)           
       \end{cases}
    \end{equation}

    \begin{equation}
        v_1(t) = \begin{cases}
            700 \ \mathrm{km/h}, \text{for} \ t \in [0,30) \\
            670 \ \mathrm{km/h}, \text{for} \ t \in [30,60)
        \end{cases}
    \end{equation}

    \item Scenario 2, aircraft diverging:  the predecessor and follower aircraft airspeeds are, respectively,
    
        \begin{equation}
        v_0(t) = \begin{cases}
            700 \ \mathrm{km/h}, \text{for} \ t \in [0,30) \\
            670 \ \mathrm{km/h}, \text{for} \ t \in [30,60)
        \end{cases}
    \end{equation}

        \begin{equation}
       v_1(t) =  \begin{cases}
           600 \ \mathrm{km/h}, \text{for} \ t \in [0,10), [20,30), [40,50) \\
           650 \ \mathrm{km/h}, \text{for} \ t \in [10,20), [30,40), [50,60)           
       \end{cases}
    \end{equation}

\end{enumerate}

Both PDW and DD metrics were computed for the two scenarios and scaled using min–max normalization, yielding values between 0 (minimum) and 1 (maximum), as shown in Fig. \ref{fig_complexity_met_comp}. The left plots present the airspeed profiles for each scenario, while the right plots depict the evolution of airspace complexity metrics over time. In Scenario 1, $v_1(t) > v_0(t)$, leading to the convergence of the two aircraft, whereas in Scenario 2, $v_1(t) < v_0(t)$, resulting in increasing separation. In both scenarios, DD produced identical values, since it only accounts for the number of aircraft that changed airspeed within each 2-min interval. In contrast, the PDW metric captured the continuous influence of both the separation and its rate of change on airspace complexity, increasing when the aircraft are converging and decreasing when they are diverging. Furthermore, the rate of separation between the aircraft was observed to affect the slope of the PDW curve, as higher separation rates lead to a steeper slope. Unlike DD, PDW reflects gradual changes and transient behaviors in pairwise aircraft interactions. Moreover, PDW highlights situations where conflicts may be emerging even before DD detects them.

\begin{figure}[ht]
\centerline{\includegraphics[scale=0.8]{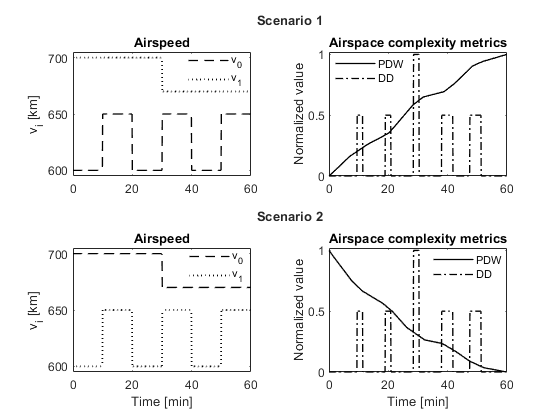}}
\caption{Comparison of airspace complexity metrics}
\label{fig_complexity_met_comp}
\end{figure}

\subsection{Validation of suboptimal solution using a shooting method} \label{subsec_Shoot}
To validate the assumption $J_{d_{i,i-1}}^*(t) = 0$ $\forall \ t \in [0, t^*_f]$, which led to the suboptimal solution (\ref{eq_OCP_subop}), a shooting method was developed to determine $J_{d_{i,i-1}}^*(0)$. Since $ J_{d_{i,i-1}}^*(t^*_f)=0$ from (\ref{eq_transv_Jd}), the problem is cast as an Initial Value Problem (IVP), where (\ref{eq_dot_x})–(\ref{eq_dot_Gamma}) are simulated for different initial estimates until these terminal conditions are met. Algorithm \ref{alg_shooting} outlines the shooting method, with $J_{d_{i,i-1}}^{*[k]}$ denoting the $k$-th iteration value of the costate $J_{d_{i,i-1}}(t)$. A flight scenario that emulates an actual operation is proposed to apply Algorithm \ref{alg_shooting}, with the parameters presented in Table \ref{tab_shoot_param}. Moreover, we consider the all-electric aircraft model $A_1$ (as per section \ref{subsec_Sim_Par}), $x_{0_i} = 0$ and $v_w = 10$ $\mathrm{km/h}$. The parameters in Algorithm \ref{alg_shooting} were set as $tol = 10^{-6}$, $\beta = 0.5$, $k_{max} = 30$.

\begin{algorithm}
\caption{Shooting method for costate approximation} \label{alg_shooting}
\begin{algorithmic}

\AlgBlock{1. Initialization}{}
    \State \quad (a) Provide an estimate for $v_i^* = v^*_{i,est}$ based on knowledge of the flight envelope of aircraft $A_i$.
    \State \quad (b) Solve $p_i = 0$ in (\ref{eq_OCP_sol}) for $J_{d_{i,i-1}}^*(0)$ with the value of $v_i^*$ previously defined and a given $v_{i-1}$.
    \State \quad (c) Initialize the parameter $k = 0$.
    \State \quad (d) Define a tolerance $tol$, a maximum number of iterations $k_{max}$, and step size $\beta$.

\AlgBlock{2. Do}{}
    \State \quad (a) Simulate the system (\ref{eq_dot_x})-(\ref{eq_dot_Gamma}) until $x_i(t^*_{f})=x_{{f}_i}$  using a numerical method that solves ordinary differential equations. The value of $v_i^*$ is computed by setting $p_i = 0$ in (\ref{eq_OCP_sol}) using a root solver for nonlinear functions.
    \State \quad (b) Compute error $\epsilon_1 = J_{d_{i,i-1}}^{*[k]}(t^*_f) - J_{d_{i,i-1}}^*(t^*_f)$.    
    \State \quad (c) Update $J_{d_{i,i-1}}^{*[k+1]}(0) = J_{d_{i,i-1}}^{*[k]}(0) - \beta \epsilon_1$.    
    \State \quad (d) Update $k = k + 1$.

\AlgBlock{3. Termination condition}{Stop if $|\epsilon_1| < tol$ or $k \geq k_{max}$.}

\end{algorithmic}
\end{algorithm}

\begin{table} [hbt]
\centering
   \caption{Shooting Method Validation Parameters}\label{tab_shoot_param}
   \centering
   \begin{tabular}{|c |c | c| c| }
   \hline
     \textbf{Parameter} & \textbf{Value} & \textbf{Parameter} & \textbf{Value} \\
     \hline
         $x_{f_i}$ $[\mathrm{km}]$  & 160 & $h_c$ $[\mathrm{km}]$  & 1 \\
         $\alpha$ $[\mathrm{C.m/s}]$ & 10$^4$ & $CI_i$ $[\mathrm{C/s}]$ & 0.2  \\
         $\rho$ $[\mathrm{kg/m^3}]$  & 1.112 & $v_{i-1}$ $[\mathrm{km/h}]$ & 135 \\
     \hline
     \end{tabular}
\end{table}

Fig. \ref{fig_shooting_AEA} shows the results of the proposed shooting method for different estimates of the aircraft airspeed $v^*_{i,est}$ within its flight envelope. We observe that the initial value of the costate $J_{d_{i,i-1}}^*$ converges rapidly to zero, which is also the value of  $J_{d_{i,i-1}}^*(t^*_f)$. This justifies and validates the assumption that $J_{d_{i,i-1}}^*(t) = J_{d_{i,i-1}}^*(t^*_f)=0$ $\forall \ t \in [0, t^*_f]$.

\begin{figure}[ht!]
\centerline{\includegraphics[scale=0.7]{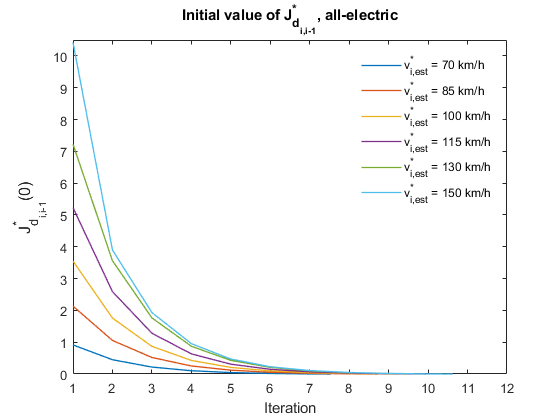}}
\caption{Shooting Method Result, all-electric aircraft}
\label{fig_shooting_AEA}
\end{figure}

\subsection{Application of the proposed methodology in a simulated flight scenario} \label{subsec_Scen}

This section validates the proposed unified framework for computing the airspeed of all-electric aircraft operating in string-stable predecessor–follower platoons under longitudinal wind conditions. As noted in Section \ref{sec_RelatedW}, the literature provides limited studies on string stability in aircraft platoons. Table \ref{table_ststab_comparison} presents a comparison between the proposed approach and existing work. The validation is conducted using a platoon of three all-electric aircraft operating in a constant-altitude air corridor (Fig. \ref{fig_AEA_corridor}).

\begin{figure}[ht!]
\centerline{\includegraphics[scale=0.3]{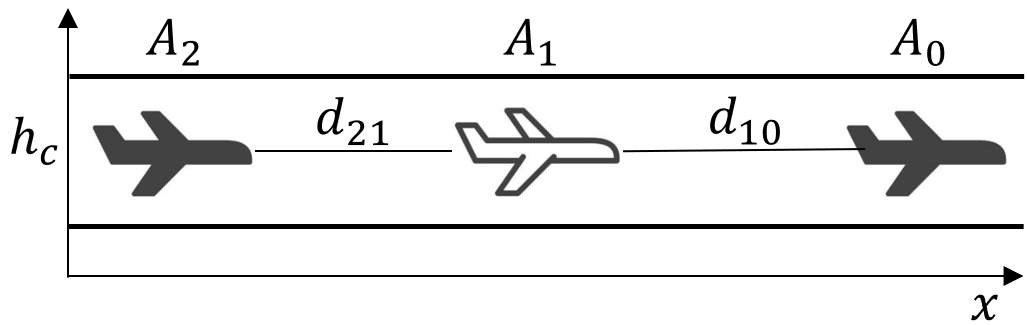}}
\caption{All-electric aircraft platoon in an AAM corridor}
\label{fig_AEA_corridor}
\end{figure}

\begin{table} [ht!]
   \caption{Comparison of String-stable aircraft procedures}\label{table_ststab_comparison}
   \centering
   \begin{tabular}{|c| c c c c c c|}
   \hline
   \textbf{Parameter} & \textbf{\cite{Allen2002}} & \textbf{\cite{Swieringa2015}} & \textbf{\cite{Weitz2018}} & \textbf{\cite{Riehl2022}} & \textbf{\cite{Stephens2023}} & \textbf{Proposed} \\
   \hline
   Applicable to all-electric aircraft & \ding{55} & \ding{55} & \ding{55} & \ding{55} & \ding{55} & \checkmark \\
   \hline
   Nonlinear aircraft dynamics & \ding{55} & \ding{55} & \ding{55} & \ding{55} & \checkmark & \checkmark \\
   \hline
   Wind disturbance & \ding{55} & \checkmark & \ding{55} & \checkmark & \checkmark & \checkmark \\
   \hline
   Operational costs & \checkmark & \ding{55} & \ding{55} & \checkmark & \ding{55} & \checkmark \\
   \hline
   Airspace complexity (ATC workload) & \ding{55} & \checkmark & \checkmark & \ding{55} & \ding{55} & \checkmark \\
   \hline
   Optimality guarantees & \ding{55} & \ding{55} & \ding{55} & \ding{55} & \ding{55} & \checkmark \\
   \hline
   Heterogeneous platoons & \ding{55} & \checkmark & \ding{55} & \ding{55} & \checkmark & \checkmark \\
     \hline
     \end{tabular}
\end{table}

Two different aircraft models are considered to represent heterogeneous platoon parameters. The aircraft are denoted as $A_0$ (leader), $A_1$ (first follower), and $A_2$ (second follower), with parameters provided in Table \ref{tab_sim_param}. Three scenarios are analyzed (Table \ref{tab_sim_AEA_CI}), each with different cost indices, resulting in distinct airspeed profiles and separation trajectories. Additional trajectory parameters are listed in Table \ref{tab_sim_AEA_plat_param}, where subscripts indicate the aircraft’s position in the platoon. The airspeed of each follower aircraft was computed by \eqref{eq_suboptimal_cases}. For each simulation scenario, the resulting airspeed profiles are shown in Figs. \ref{fig_resul_scenarios}(a), (d) and (g), the aircraft separation trajectories are presented in Figs. \ref{fig_resul_scenarios}(b), (e) and (h) and the battery charge depletion curves are depicted in Figs. \ref{fig_resul_scenarios}(c), (f) and (i). In all cases, the aircraft separation remained above the defined minimum value during the entire procedure.

\begin{table} [t!]
\centering
   \caption{Air corridor operation of all-electric aircraft}\label{tab_sim_AEA_plat_param}
   \centering
   \begin{tabular}{|c |c | c | c |}
   \hline
     \textbf{Parameter} & \textbf{Value} & \textbf{Parameter} & \textbf{Value} \\
     \hline
        $x_{0_0}$ $[\mathrm{km}]$ & 10 & $x_{f_0}$ $[\mathrm{km}]$ & 80 \\
        $x_{0_1}$ $[\mathrm{km}]$ & 6 & $x_{f_1}$ $[\mathrm{km}]$ & 74 \\
        $x_{0_2}$ $[\mathrm{km}]$ & 4 & $x_{f_2}$ $[\mathrm{km}]$ & 68 \\
        $\rho$ $[\mathrm{kg/m^3}]$ & 1.112 &  $h_c$ $[\mathrm{km}]$ & 1 \\
        $v_w$ $[\mathrm{km/h}]$ & 0 & $\alpha$ $[\mathrm{C.m/s}]$ & 10$^4$ \\
        $v_{0}$ $[\mathrm{km/h}]$ & 100 & $d_{min}$ $[\mathrm{km}]$ & 1 \\ 
     \hline
     \end{tabular}
\end{table}

\begin{table} [t!]
\centering
   \caption{Simulated flight scenarios}\label{tab_sim_AEA_CI}
   \centering
   \begin{tabular}{|c |c | c |}
   \hline
     \textbf{Scenario} & \textbf{$CI_1$} & \textbf{$CI_2$} \\
     \hline
        1 & 90 & 30 \\
        2 & 100 & 35 \\
        3 & 70 & 40 \\
     \hline
     \end{tabular}
\end{table}

\begin{figure}[!h]
\centerline{\includegraphics[scale=0.6]{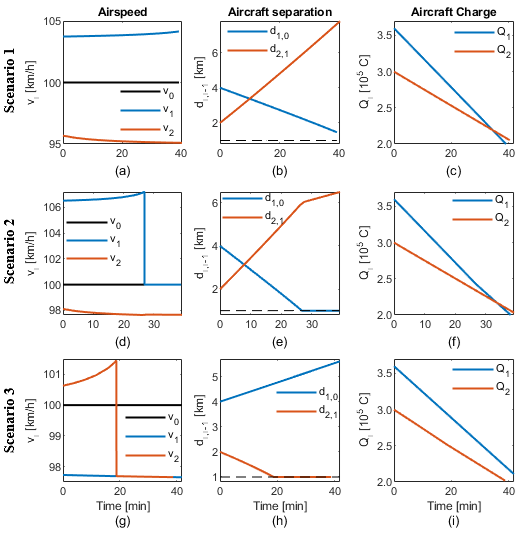}}
\caption{Simulated flight scenario for different cost indices}
\label{fig_resul_scenarios}
\end{figure}

\subsection {Effect of longitudinal wind}

The effect of longitudinal wind on the computed airspeed is illustrated in Fig. \ref{fig_resul_wind} for different aircraft separation distances (Fig. \ref{fig_resul_wind}(a)), scaling factor values $\alpha$ (Fig. \ref{fig_resul_wind}(b)), and cost indices (Fig. \ref{fig_resul_wind}(c)). The results indicate that stronger headwinds ($v_w < 0$) lead to higher values of $v_i$ to compensate for their effect, whereas stronger tailwinds ($v_w > 0$) result in lower values of $v_i$, as the wind favors the aircraft motion. These trends are consistent with the expected behavior based on physical observation.

\begin{figure}[H]
\centerline{\includegraphics[scale=0.6]{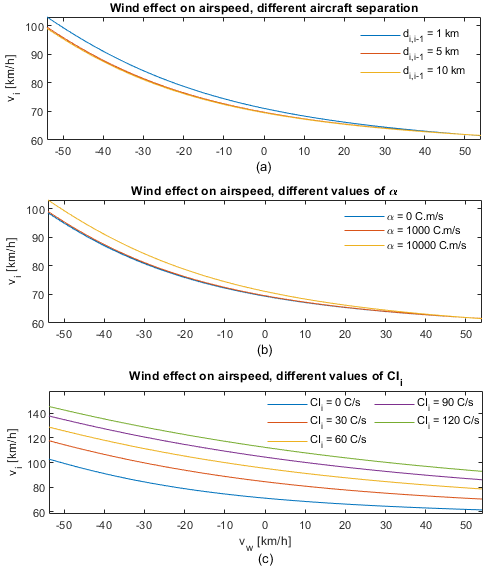}}
\caption{Effect of the wind in the airspeed computation, where $v_w > 0$ denotes tailwind and $v_w < 0$ denotes headwind}
\label{fig_resul_wind}
\end{figure}

\subsection {String stability}

String stability is assessed for the three simulated scenarios. The string stability coefficient was computed for each aircraft pair in each scenario using \eqref{eq_SS_cond_applied}, and the corresponding maximum values are reported in Table \ref{table_SS_coef}. Since $K_{i,i-1}(t) < 1 \ \forall t \in [0,t_f]$ in all cases, the platoon is string stable in every scenario.

\begin{table} [ht]
\centering
   \caption{Maximum value of the string stability coefficient}\label{table_SS_coef}
   \centering
   \begin{tabular}{|c |c | c |}
   \hline
     \textbf{Scenario} & \textbf{$K_{1,0}(t)$} & \textbf{$K_{2,1}(t)$} \\
     \hline
        1 & 1.3 $\mathrm{x}10^{-2}$ & 1.7 $\mathrm{x}10^{-2}$  \\
        2 & 2.0 $\mathrm{x}10^{-2}$  & 1.6 $\mathrm{x}10^{-2}$   \\
        3 & 5.4 $\mathrm{x}10^{-3}$ & 3.3 $\mathrm{x}10^{-2}$ \\
     \hline
     \end{tabular}
\end{table}
Additionally, the scenarios were evaluated under an airspeed disturbance defined as
\begin{equation}
\delta v_{i-1}(t) = \bar{w}\sin\biggl(\frac{\pi(t-t_0)}{\Delta t}\biggl), t \in [t_0,t_0+\Delta t]
\end{equation}
with the disturbance amplitude $\bar{w}=18$ $\mathrm{km/h}$, initial time $t_0=8.3$ $\mathrm{min}$ and duration of $\Delta t=1.7$ $\mathrm{min}$. Fig. \ref{fig_resul_string_st_sc} depicts the airspeed error $v_i-\bar{v}_{i-1}$ for the three scenarios. The amplitude of the airspeed error for the second pair ($A_2,A_1$) remains consistently smaller than that of the first pair ($A_1,A_0$), indicating attenuation of the disturbance as it propagates along the platoon and thereby confirming string stability.

The influence of longitudinal wind and the scaling parameter $\alpha$ on the string stability coefficient is analyzed in Fig. \ref{fig_resul_K_sc}. As shown in Fig. \ref{fig_resul_K_sc}(a), increasing headwind intensity leads to larger values of $K_{i,i-1}$, resulting in stronger amplification of disturbances in the predecessor aircraft’s airspeed as they propagate along the platoon. The effect of the scaling parameter $\alpha$ is presented in Fig. \ref{fig_resul_K_sc}(b), where larger values of $\alpha$ also produce higher values of $K_{i,i-1}$, indicating that more complex airspace configurations lead to stronger propagation of disturbances along the platoon.

\begin{figure}[H]
\centerline{\includegraphics[scale=0.68]{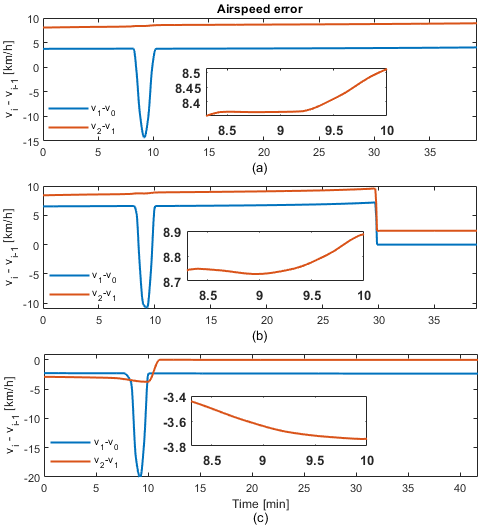}}
\caption{Airspeed error caused by a disturbance in the leader aircraft’s airspeed, for each scenario
}
\label{fig_resul_string_st_sc}
\end{figure}

\begin{figure}[H]
\centerline{\includegraphics[scale=0.7]{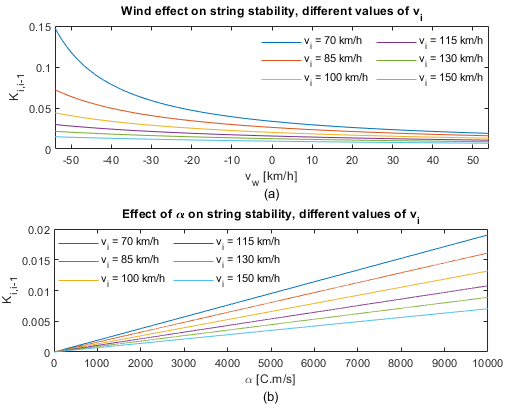}}
\caption{Effect of wind and $\alpha$ on string stability}
\label{fig_resul_K_sc}
\end{figure}

\section{Conclusions}\label{sec_Conc}
This paper presents a novel optimal control framework to compute the airspeeds of all-electric aircraft platoons operating under predecessor–follower procedures in the presence of longitudinal wind. The proposed framework addresses a class of airspace coordination problems involving pairwise aircraft interactions, for which a novel cost function is introduced to perform a trade-off between aircraft operating costs and the complexity of these airspace procedures, which is modeled by the proposed pairwise dynamic workload (PDW) metric. A closed-form suboptimal airspeed solution is derived, together with a sufficient condition that ensures string stability of the platoon. Simulation results demonstrate the application of the proposed methodology to constant-altitude air corridor operations of all-electric aircraft, where both minimum separation and string stability are maintained throughout the procedure. The influence of longitudinal wind on follower aircraft airspeed is evaluated, along with the impact of wind and the airspace complexity scaling parameter on the string stability coefficient. Overall, the proposed framework provides a foundation for autonomous and supervised operations of all-electric aircraft, supporting low-altitude procedures in uncontrolled airspace and facilitating the integration of AAM/UAM concepts within the broader air traffic management system.

\section*{Acknowledgments}
The authors would like to acknowledge the financial support of the Fonds de Recherche du Québec – Nature et Technologies (FRQNT) and the Natural Sciences and Engineering Research Council of Canada (NSERC).

\section*{Data availability}
All data generated or analyzed during this study are included in the manuscript.

\section*{Funding}
Funding for this project is provided by the Fonds de Recherche du Québec – Nature et Technologies (FRQNT), grant number: 366655.

\section*{Author contribution}
Lucas Souza e Silva developed the methodology, performed the analysis, and prepared the original draft. Luis Rodrigues contributed to the conceptual development of the study, provided supervision and technical guidance, and reviewed and edited the manuscript. Both authors approved the final version of the manuscript.

\section*{Declarations}
Conflict of interest: the authors have no relevant financial or non-financial interests to disclose.

\bibliography{FIM_ATC_bibliography}

@misc{NYT,
  title={{Drunk and Asleep on the Job: Air Traffic Controllers Pushed to the Brink}},
  author={Emily Steel and Sydney Ember},
  url={https://www.nytimes.com/2023/12/02/business/air-traffic-controllers-safety.html}
}

@misc{FAA_FIM,
  title={{ADS-B In Interval Management}},
  author={{Federal Aviation Administration}},
}

@misc{Kratos,
  title={{Kratos Unmanned Systems Demonstrates Leader Follower Platooning Technology in Quebec's Forestry Sector in a Timber Hauling Operation}},
  author={{GlobeNewswire}},
  url={https://www.globenewswire.com/news-release/2025/01/27/3015553/224/en/Kratos-Unmanned-Systems-Demonstrates-Leader-Follower-Platooning-Technology-in-Quebec-s-Forestry-Sector-in-a-Timber-Hauling-Operation.html}
}

@article{Barzkar_Ghassemi2022,
author = {A. Barzkar and M. Ghassemi},
title = {Components of Electrical Power Systems in More and All-Electric Aircraft: A Review},
journal = {IEEE Transactions on Transportation Electrification},
volume = {8},
number = {4},
year = {2022}
}

@article{Path,
author = {Steven E. Shladover},
title = {{PATH at 20—History and Major Milestones}},
journal = {IEEE Transactions on Intelligent Transportation Systems},
volume = {8},
number = {4},
year = {2007}
}

@inproceedings{SouzaeSilva2025,
author = {Lucas Souza e Silva and Ali Akgunduz and Luis Rodrigues},
title = {{A Unified Approach for Optimal Cruise Airspeed with Variable Cost Index}},
booktitle = {International Conference on Control, Decision and Information Technologies},
month = {10},
year = {2025},
url = {}
}

@article{Hartl1995,
author = {R. F. Hartl and Suresh Sethi and Raymond Vickson},
title = {{A Survey of the Maximum Principles for Optimal Control Problems with State Constraints}},
journal = {SIAM Review},
volume = {37},
number = {2},
pages = {181--218},
year = {1995}
}

@inproceedings{Pawlak1996,
author = {William Pawlak and Christopher Brinton and Kimberly Crouch and Kenneth Lancaster},
title = {{A framework for the evaluation of air traffic control complexity}},
booktitle = {{AIAA Guidance, Navigation, and Control Conference}},
address = {San Diego, USA},
month = {07},
year = {1996},
url = {https://doi.org/10.2514/6.1996-3856}
}

@misc{Laudeman1998,
  title={{Dynamic Density: An Air Traffic Management Metric}},
  author={I. Laudeman and S. Shelden and R. Brannstrom and C. Brasil},
  url={https://ntrs.nasa.gov/citations/19980210764}
}

@inproceedings{Chatterji2001,
author = {Gano B. Chatterji and Banavar Sridhar},
title = {{Measures for Air Traffic Controller Workload Prediction}},
booktitle = {{AIAA Aircraft Technology, Integration and Operations Forum}},
address = {Los Angeles, USA},
month = {10},
year = {2001},
url = {https://doi.org/10.2514/6.2001-5242}
}

@article{Turri2017,
author = {Valerio Turri and Bart Besselink and Karl H. Johansson},
title = {Cooperative Look-Ahead Control for Fuel-Efficient and Safe Heavy-Duty Vehicle Platooning},
journal = {IEE Transactions on Control Systems Technology},
volume = {25},
number = {1},
pages = {12--28},
doi = {10.1109/TCST.2016.2542044},
year = {2017}
}

@article{Hussein2021,
author = {Ahmed A. Hussein and Hesham A. Rakha},
title = {Vehicle Platooning Impact on Drag Coefficients and
Energy/Fuel Saving Implications},
journal = {IEE Transactions Vehicular Technology},
volume = {71},
number = {2},
pages = {1199 -- 1208},
doi = {10.1109/TVT.2021.3131305},
year = {2021}
}

@article{Giulietti2000,
author = {Fabrizio Giulietti and Lorenzo
Pollini and Mario Innocenti},
title = {Autonomous Formation Flights},
journal = {{IEEE Control Systems Magazine}},
volume = {20},
number = {6},
pages = {34--44},
doi = {10.1109/37.887447},
year = {2000}
}

@article{Brittain2019,
author = {Marc Brittain and Peng Wei},
title = {Autonomous Air Traffic Controller: A Deep Multi-Agent Reinforcement Learning Approach},
journal = {{arxiv}},
doi = {10.48550/arXiv.1905.01303},
year = {2019}
}

@inproceedings{Su2024,
author = {Yu-Hsiang Su and Parijat Bhowmick and Alexander Lanzon},
title = {{A Negative Imaginary Solution to an Aircraft Platooning Problem}},
booktitle = {2024 European Control Conference},
address = {Stockholm, Sweden},
month = {06},
year = {2024},
}

@article{Tomlin1998,
author = {C. Tomlin and G.J. Pappas and S. Sastry},
title = {Conflict resolution for air traffic management: a study in multiagent hybrid systems},
journal = {{IEEE Transactions on Automatic Control}},
volume = {43},
number = {4},
pages = {509--521},
year = {1998}
}

@inproceedings{Gorodetsky2008,
author = {V. Gorodetsky and O. Karsaev and V. Samoylov and V. Skormin},
title = {{Multi-AgentTechnology for Air Traffic Control and Incident Management in AirportAirspace}},
booktitle = {International Workshop on Agents in Traffic andTransportation},
pages = {119--125},
year = {2008},
}

@misc{Wolfe2009,
  title={{A Multiagent Simulation of Collaborative Air Traffic Flow Management}},
  author={Shawn R. Wolfe and Peter A. Jarvis and Francis Y. Enomoto and Maarten Sierhuis and Bart-Jan van Putten},
booktitle = {{NASA Ames Research Center}},
   year = {2009}
}

@article{Toratani2022,
author = {Daichi Toratani and Takayuki Yoshihara and Atsushi Senoguchi},
title = {Support algorithm for air traffic controllers’ arrival spacing: Improvement of
trajectory estimation using Gaussian Process Regression},
journal = {Control Engineering Practice},
volume = {128},
doi = {10.1016/j.conengprac.2022.105343},
year = {2022}
}

@article{Yang2018,
author = {Yi Yang and Ying Nan and Ming Tong},
title = {Cooperative Route Planning for Multiple Aircraft in a Semifree ATC System},
journal = {Mathematical Problems in Engineering},
doi = {10.1155/2018/3521905},
year = {2018}
}

@inproceedings{Chen2015,
author = {Mo Chen and Qie Hu and Casey Mackin and Jaime F. Fisac and Claire J. Tomlin},
title = {{Safe platooning of unmanned aerial vehicles via reachability}},
booktitle = {54th IEEE Conference on Decision and Control (CDC)},
address = {Osaka, Japan},
doi = {10.1109/CDC.2015.7402951}, 
month = {12},
year = {2015},
}

@inproceedings{Su2023,
author = {Yu-Hsiang Su and Parijat Bhowmick and Alexander Lanzon},
title = {{Cooperative Control of Multi-Agent Negative Imaginary Systems with Applications to UAVs, Including Hardware Implementation Results}},
booktitle = {2023 European Control Conference (ECC)},
address = {Bucharest, Romania},
doi = {10.23919/ECC57647.2023.10178371}, 
month = {06},
year = {2023},
}

@article{Mayle2023,
author = {Melody N. Mayle and Rajnikant Sharma},
title = {Platooning in UAM airspace
structures: applying trajectory
shaping guidance law and
exploiting cooperative
localization},
journal = {Frontiers in Aerospace Engineering},
doi = {10.3389/fpace.2023.1176812},
year = {2023}
}

@inproceedings{Buzogany1993,
author = {L. Buzogany and M. Pachter and J. D'azzo},
title = {{Automated control of aircraft in formation flight}},
booktitle = {Guidance, Navigation and Control Conference},
address = {Monterey, USA},
doi = {10.2514/6.1993-3852}, 
month = {08},
year = {1993},
}

@inproceedings{Khoshdel2024,
author = {Sahand Khoshdel and Fatemeh Afghah and Qi Luo},
title = {{SkyGrid: Energy-Flow Optimization at Harmonized Aerial Intersections}},
booktitle = {100th IEEE Vehicular Technology Conference},
address = {Washington, DC, USA},
doi = {10.1109/VTC2024-Fall63153.2024.10757812}, 
month = {10},
year = {2024},
}

@inproceedings{Ribichini2003,
author = {G. Ribichini and E. Frazzoli},
title = {{Efficient coordination of multiple aircraft systems}},
booktitle = {42nd IEEE International Conference on Decision and Control},
pages = {1035–-1040},
address = {Maui, HI, USA},
doi = {10.1109/CDC.2003.1272704}, 
month = {12},
year = {2003},
}

@article{Benvenuti2024,
author = {Luca Benvenuti and Alberto de Santis},
title = {A Design Methodology for Commercial Aircraft Formation Flight Plans With Minimal Fuel Consumption},
journal = {IEEE Transactions on Aerospace and Electronic Systems},
volume = {60},
number = {4},
pages = {4157--4169},
doi = {10.1109/TAES.2024.3375838},
year = {2024}
}

@article{Hartjes2018,
author = {Sander Hartjes and Marco E. G. van Hellenberg Hubar  and Hendrikus G. Visser},
title = {Multiple-phase trajectory optimization for formation flight in civil aviation},
journal = {CEAS Aeronautical Journal},
volume = {10},
pages = {453–-462},
doi = {10.1007/s13272-018-0329-9},
year = {2018}
}

@article{Tang2022,
author = {Xinmin Tang and Xiaona Lu and Pengcheng Zheng},
title = {Aircraft Autonomous Separation Assurance Based on Cooperative Game Theory},
journal = {MDPI Aerospace},
doi = {10.3390/aerospace9080421},
year = {2022}
}

@article{Curtis2025,
author = {Olivia Curtis and Yibing Xie and Man Liang and Cees Bil},
title = {Energy-Efficient Path Planning for Commercial Aircraft Formation Flights},
journal = {MDPI Engineering Proceedings},
doi = {10.3390/engproc2024080015},
year = {2025}
}

@article{Hartjes2019,
author = {Sander Hartjes and Hendrikus G. Visser and Marco E. G. van Hellenberg Hubar},
title = {Trajectory Optimization of Extended Formation Flights for Commercial Aviation},
journal = {MDPI Aerospace},
doi = {10.3390/aerospace6090100},
year = {2019}
}

@article{Cerezo2021,
author = {María Cerezo-Magaña and Alberto Olivares and Ernesto Staffetti},
title = {Formation Mission Design for Commercial Aircraft Using Switched Optimal Control Techniques},
journal = {IEEE Transactions on Aerospace and Electronic Systems},
volume = {57},
number = {4},
pages = {2540--2557},
doi = {10.1109/TAES.2021.3061790},
year = {2021}
}

@article{Cerezo2022,
author = {María Cerezo-Magaña and Alberto Olivares and Ernesto Staffetti},
title = {A Stochastic Switched Optimal Control Approach to Formation Mission Design for Commercial Aircraft},
journal = {IEEE Transactions on Aerospace and Electronic Systems},
volume = {58},
number = {5},
pages = {4342--4360},
doi = {10.1109/TAES.2022.3161890},
year = {2022}
}

@misc{FAA_Conops,
author = {FAA},
title = {Urban Air Mobility (UAM) Concepts of Operation v2.0},
url = {https://www.faa.gov/sites/faa.gov/files/Urban20Air20Mobility2028UAM2920Concept20of20Operations202.0_0.pdf},
year = {2023}
}

@article{Feng2019,
author = {Shuo Feng and Yi Zhang and Shengbo Eben Li and Zhong Cao and Henry X. Liu and Li Li},
title = {String stability for vehicular platoon control: Definitions and analysis methods},
journal = {Annual Reviews in Control},
volume = {47},
pages = {81--97},
doi = {10.1016/j.arcontrol.2019.03.001},
year = {2019}
}

@article{Stephens2023,
author = {Shawn S. Stephens and David W. Casbeer and Dzung Tran and Donald L. Kunz},
title = {String Stable Heterogeneous Aircraft in a Windy Environment},
journal = {Journal of Guidance, Control, and Dynamics},
volume = {46},
number = {6},
pages = {1052--1065},
doi = {10.2514/1.G006391},
year = {2023}
}

@article{Besselink2017,
author = {B. Besselink and K. H. Johansson},
title = {String Stability and a Delay-Based
Spacing Policy for Vehicle Platoons Subject to Disturbances},
journal = {IEEE Transactions on Automatic Control},
volume = {62},
number = {9},
pages = {4376--4391},
year = {2017}
}

@inproceedings{Allen2002,
author = {Michael Allen and Jack Ryan and Curtis Hanson and James Parle},
title = {{String Stability of a Linear Formation Flight Control System}},
booktitle = {AIAA Guidance, Navigation, and Control Conference and Exhibit},
address = {Monterey, California, USA},
month = {08},
year = {2002},
doi = {10.2514/6.2002-4756}
}

@inproceedings{Swieringa2015,
author = {Kurt A. Swieringa},
title = {{The String Stability of a Trajectory-Based Interval Management Algorithm in the Midterm Airspace}},
booktitle = {15th AIAA Aviation Technology, Integration, and Operations Conference},
address = {Dallas, Texas, USA},
month = {06},
year = {2015},
doi = {10.2514/6.2015-2278}
}

@inproceedings{Weitz2018,
author = {Lesley A. Weitz and Kurt A. Swieringa},
title = {{Comparing Interval Management Control Laws for Steady-State Errors and String Stability}},
booktitle = {AIAA Guidance, Navigation, and Control Conference},
address = {Kissimmee, Florida, USA},
month = {01},
year = {2018},
doi = {10.2514/6.2018-1340}
}

@misc{E430,
author = {{GreenWing International}},
title = {E430 electric aircraft},
url = {https://greenwing.aero/?page id=2345},
year = {2022}
}

@misc{Velis,
  title={{Electric Pioneer - Velis Electro}},
  author={{Pipistrel}},
  url={www.pipistrel-aircraft.com/products/velis-electro/}
}

@book{Anderson1999,
  title={Aircraft Performance and Design},
  author={John D. Anderson},
  year={1999},
  publisher={WCB/McGraw-Hill}
}

@inproceedings{LiRodrigues2023,
author = {Steven Li and Luis Rodrigues},
title = {{Optimal Cruise Airspeed for Hybrid-Electric and Electric Aircraft: Applications to Air Mobility}},
booktitle = {{Conference on Control Technology and Applications (CCTA)}},
address = {Bridgetown, Barbados},
month = {08},
year = {2023},
doi = {10.1109/CCTA54093.2023.10252879}
}

@article{Riehl2022,
author = {James R. Riehl and Esteban A. L. Hufstedler and Philippe Chatelain and Julien M. Hendrickx},
title = {String Stability of Energy-Saving Aircraft Formations},
journal = {Journal of Guidance, Control, and Dynamics},
volume = {45},
number = {5},
pages = {935--943},
doi = {10.2514/1.G006207},
year = {2022}
}

@book{Pontr,
  title={The Mathematical Theory of Optimal Processes},
  author={Pontryagin, L S and Boltyanskii, V G and Gamkrelidze, R V and Mishchenko, E F},
  year={1962},
  publisher={Interscience Publishers}
}

\end{document}